\documentclass[11pt,a4paper]{article}
\usepackage[margin=2cm]{geometry}
\usepackage[T1]{fontenc}
\usepackage[utf8]{inputenc}
\usepackage[noadjust]{cite}
\usepackage{lmodern}
\usepackage{amsthm}
\usepackage{thmtools}
\makeatletter
\@ifundefined{newcounteralias}{}{%
  \renewcommand\thmt@autorefsetup{%
    \@xa\def\csname\thmt@envname autorefname\@xa\endcsname
      \@xa{\thmt@thmname}%
  }%
}
\makeatother

\usepackage{thm-autoref}
\usepackage[shortlabels]{enumitem}
\usepackage[center]{caption}
\usepackage{latexsym}
\usepackage{amsfonts}
\usepackage{graphicx}
\usepackage{mathrsfs}
\usepackage[normalem]{ulem}
\usepackage{amssymb, mathdots, mathtools, xkeyval, tikz, tkz-euclide, adjustbox, float, array, subcaption, listings}
\usepackage{comment}
\usepackage{xspace}
\usepackage[colorlinks=true,
linkcolor =blue,
citecolor=blue,
urlcolor=blue]{hyperref}
\usepackage[capitalise]{cleveref}

\theoremstyle{plain}

\newtheorem{theorem}{Theorem}[section]
\newtheorem{lemma}[theorem]{Lemma}

\newtheorem{corollary}[theorem]{Corollary}

\newtheorem{question}[theorem]{Question}

\newtheorem{claim}{Claim}[theorem]
\newenvironment{claimproof}{\noindent\textit{Proof of Claim \theclaim:}}{\leavevmode\unskip\penalty9999\hbox{}\nobreak\hfill\quad\hbox{$\lrcorner$}\par} 

\crefname{claim}{Claim}{Claims}
\crefname{question}{Question}{Questions}

\newcommand{\lcol}[1]{\textsc{List #1-Coloring}\xspace}
\newcommand{\MWIS}{\textsc{MWIS}\xspace}
\newcommand{\NP}{\textsf{NP}\xspace}
\newcommand{\Oh}{\mathcal{O}}

\newcommand{\cc}{\textsf{cc}}

\newcommand{\w}{\textsf{w}}
\newcommand{\m}{\textsf{m}}
\newcommand{\dist}{\mathsf{dist}}
\renewcommand{\phi}{\varphi}

\newcommand{\cC}{\mathcal{C}}

\newcommand{\cG}{\mathcal{G}}

\newcommand{\cI}{\mathcal{I}}
\newcommand{\cJ}{\mathcal{J}}

\newcommand{\cS}{\mathcal{S}}

\usepackage{todonotes}

\title{List 3-Coloring in tame classes of graphs in quasipolynomial time}

\author{Paweł Rzążewski\thanks{Warsaw University of Technology. Supported by the National Science Centre grant 2024/54/E/ST6/00094.}}
\date{}

\hypersetup{hypertexnames=false}
\begin{document}
	\begin{titlepage}
	\date{}
	\maketitle

\begin{abstract}
    A class of graphs $\mathcal{G}$ is \emph{tame} if there is a constant $c$ such that every $n$-vertex graph in $\mathcal{G}$ has at most $n^c$ minimal separators.
    Many important graph classes are tame, including chordal graphs, permutation graphs, and circle graphs.
    Several fundamental problems, most notably \textsc{Max Weight Independent Set}, can be solved in polynomial time on every tame class, which explains in a unified way their tractability on all these classes.

    We show that \textsc{List 3-Coloring} can be solved in quasipolynomial time on every tame class of graphs.
\end{abstract}

	\thispagestyle{empty}

  \bigskip
  \bigskip
  \bigskip
\tableofcontents
\end{titlepage}
\thispagestyle{empty}
\newpage\setcounter{page}{1}

\section{Introduction}

A central theme in algorithmic graph theory is to understand how the complexity of computationally hard problems changes when the input graph is restricted to a specific graph class.
For many classical problems, this has led to a large number of results, each exploiting the structure of a particular class.
It is therefore very desirable to have unifying frameworks that explain the tractability of a problem in many classes at once.

One such framework is given by the notion of tame classes.
Following the nomenclature of~\cite{MilanicP21,GartlandL23I,GartlandL23II}, a class $\cG$ of graphs is \emph{tame} (resp., \emph{quasitame}) if there is a constant $c$ (resp., two constants $a$ and $c$) such that every $n$-vertex graph in $\cG$ has at most $n^c$ (resp., $2^{c \log^a n}$) minimal separators.
Understanding which graph classes are (quasi)tame has recently attracted significant attention~\cite{MilanicP21,AbrishamiCDTTV22,GartlandL23I,GartlandL23II,GajarskyJLNPRS22}.
Many important graph classes are tame, including chordal graphs~\cite{Rose70}, weakly chordal graphs~\cite{BouchitteT01}, circular-arc graphs~\cite{KloksKW98}, and circle graphs~\cite{KloksKW98}, which in turn include permutation graphs.

Bouchitt\'e and Todinca~\cite{BouchitteT01,BouchitteT02} introduced a general framework based on so-called potential maximal cliques.
It was used to solve \textsc{Max Weight Independent Set} (\MWIS)~\cite{FominV10}, and later extended by Fomin, Todinca, and Villanger~\cite{FominTV15} to the $(\mathsf{tw} \leq r, \psi)$-\MWIS problem.
Here, $r$ is a fixed integer and $\psi$ is a fixed formula of counting monadic second-order logic (CMSO$_2$), and, given a graph $G$ with weights on vertices, we look for a maximum-weight set $X \subseteq V(G)$ such that $G[X]$ has treewidth at most $r$ and satisfies $\psi$.\footnote{The actual formalism is slightly more general, see~\cite{FominTV15}.}
This captures many problems, for example \MWIS (for $r = 0$) and \textsc{Feedback Vertex Set} (for $r = 1$), as the complement of a minimum-weight feedback vertex set is a maximum-weight induced forest.

\begin{theorem}[\cite{BouchitteT01,BouchitteT02,FominTV15}]\label{thm:pmc}
    For every fixed integer $r$ and CMSO$_2$ formula $\psi$, the $(\mathsf{tw} \leq r, \psi)$-\MWIS problem can be solved in time polynomial in the number of vertices and the number of minimal separators of the input graph.
\end{theorem}

Consequently, $(\mathsf{tw} \leq r, \psi)$-\MWIS can be solved in polynomial time on every tame class, and in quasipolynomial time on every quasitame class.
Are there any other natural problems that can be solved in polynomial time on every tame class?

A hint comes from the study of $P_t$-free graphs, i.e., graphs that do not contain the $t$-vertex path $P_t$ as an induced subgraph.
Already $P_5$-free graphs are not tame: the graph consisting of two $n$-vertex cliques with a perfect matching between them is $P_5$-free and has $2^n-2$ minimal separators~\cite{PilipczukR25}.
Nevertheless, \MWIS and, more generally, $(\mathsf{tw} \leq r, \psi)$-\MWIS can be solved in polynomial time in $P_5$-free graphs~\cite{LokshtanovVV14,AbrishamiCPRS21} and in $P_6$-free graphs~\cite{GrzesikKPP22,ChudnovskyMPPR26}, using relaxations of the framework of potential maximal cliques.
Moreover, for every fixed $t$, these problems can be solved in quasipolynomial time in $P_t$-free graphs~\cite{GartlandL20,PilipczukPR21,GartlandLPPR21}.
While this is weaker than polynomial-time solvability, it is a strong indication that these problems are not \NP-hard in $P_t$-free graphs: otherwise, every problem in \NP could be solved in quasipolynomial time.

Another problem whose complexity in $P_t$-free graphs is similar to that of \MWIS is \lcol{3}.
Here we are given a graph $G$ and a \emph{list assignment} $L$ that assigns to every vertex $v \in V(G)$ a list $L(v) \subseteq \{1,2,3\}$ of colors, and we ask whether there is a proper coloring of $G$ that \emph{respects} $L$, i.e., every vertex receives a color from its list.
\lcol{3} can be solved in polynomial time in $P_t$-free graphs for $t \leq 7$~\cite{BonomoCMSSZ18}, and in quasipolynomial time for every fixed $t$~\cite{PilipczukPR21}.

Turning back to tame classes, \lcol{3} can be solved in polynomial time in some of them, but using class-specific arguments~\cite{JansenS97,WoegingerS01,RanderathS04,EnrightS11,BroersmaFGP13}.
Can we obtain such a result in a unified way, as for \MWIS in \cref{thm:pmc}?

\begin{restatable}{question}{qpoly}\label{q:poly}
    Can \lcol{3} be solved in polynomial time on every tame class of graphs?
\end{restatable}

We remark that analogous questions for some closely related problems have negative answers, unless $\textsf{P} = \NP$: \textsc{4-Coloring} is \NP-hard in circle graphs~\cite{Unger88}, \lcol{4} is \NP-hard in chordal bipartite graphs~\cite{HuangJP15}, which are tame as they are weakly chordal, and \textsc{Odd Cycle Transversal} is \NP-hard in $P_6$-free graphs of clique number at most $3$~\cite{DabrowskiFJPPR20,ChudnovskyKPRS21}, which form a tame class~\cite{PilipczukR25}.

Let us mention an important special case of \cref{q:poly}, namely circle graphs, i.e., intersection graphs of chords of a circle.
Unger~\cite{Unger92} claimed a polynomial-time algorithm for \textsc{3-Coloring} circle graphs, but only an extended abstract was published, and the proof was believed to be flawed.
Recently, Bachmann, Rutter, and Stumpf~\cite{BachmannRS24} exposed mistakes in the algorithm and argued that the problem should be considered open.
Very recently, Ajaykrishnan~E~S, Ganian, Lokshtanov, and Surianarayanan~\cite{ESGLS26} gave an algorithm for \lcol{3} in circle graphs with running time $n^{\Oh(\log n)}$.

Another special case of \cref{q:poly} was explicitly asked by Huang, Johnson, and Paulusma~\cite{HuangJP15}: what is the complexity of \lcol{3} in chordal bipartite graphs?

In this paper we give a partial answer to \cref{q:poly}.

\begin{theorem}\label{thm:main}
    \lcol{3} can be solved in quasipolynomial time on every quasitame class of graphs.
    For tame classes, the running time is $n^{\Oh(\log^5 n)}$, where $n$ is the number of vertices in the input graph.
\end{theorem}

In particular, \cref{thm:main} generalizes the result of~\cite{ESGLS26} to all tame classes, albeit with a worse running time, and gives a quasipolynomial-time algorithm for \lcol{3} in chordal bipartite graphs, partially answering the question of Huang et al.~\cite{HuangJP15}.

\paragraph{Technical overview.}
Let us present a high-level overview of the algorithm behind \cref{thm:main}.
Let $(G,L)$ be an instance of \lcol{3}, where $G$ has $n$ vertices and fewer than $2^k$ minimal separators; for tame classes, $k = \Oh(\log n)$.
We branch into $n^{\Oh(k^2 \log^2 n)}$ instances, obtained by deleting vertices and removing colors from lists, such that $(G,L)$ is a yes-instance if and only if one of them is, and every component of each of them has at most $n/2$ vertices.
Then we recurse on the components.

The branching scheme is inspired by the work of Gartland and Lokshtanov~\cite{GartlandL20,GartlandL23II}.
They showed that there is a set $D$ of at most $2k$ vertices such that every component of $G - N[D]$ has at most $n/2$ vertices, where $N[D]$ consists of $D$ and its neighbors~\cite{GartlandL23II}.
We guess the coloring of $D$ and delete $D$; afterwards, every remaining neighbor of $D$ has at most two colors in its list.
Repeating this and using a potential-function argument, we obtain $\Oh(\log n)$ such sets $D_1,\ldots,D_h$ with no edges between them, such that no vertex is adjacent to all of them.
Let $S_i$ be the set of neighbors of $D_i$.
By another result of Gartland and Lokshtanov~\cite{GartlandL23II}, if for all $i \neq j$ no component of $G - (S_i \cap S_j)$ contained vertices of both $S_i \setminus S_j$ and $S_j \setminus S_i$, then every component of $G$ would have at most $n/2$ vertices.

The main new ingredient is a way to ensure this condition, i.e., to destroy all paths from $D_i$ to $D_j$ that avoid $S_i \cap S_j$, by deleting vertices from the instance.
More generally, let $A$ and $B$ be sets at distance at least $3$ such that $G[A \cup B]$ has $c$ components.
The paths of length $3$ from $A$ to $B$ correspond to edges $xy$ with $x$ adjacent to $A$ and $y$ adjacent to $B$, and in every proper coloring, the vertices not colored $i$ cover all such edges, for each color $i$.
The inclusion-minimal sets covering these edges are inclusion-minimal sets separating $A$ from $B$ in an induced subgraph of $G$, so each of them is the union of fewer than $c$ minimal separators, and there are at most $2^{kc}$ of them.
Guessing one such set for each color and removing the color from its vertices, we ensure that one endpoint of each such edge is left with at most one color, and can be deleted.
This increases the distance between $A$ and $B$ to $4$, and a similar step increases it to $5$.
Then no vertex is at distance at most $2$ from both $A$ and $B$, which allows us to recurse on a problem where only half of the vertices may serve as internal vertices of the paths.

\section{Preliminaries}

All logarithms are binary, and for a positive integer $n$, we write $[n] = \{1,\ldots,n\}$.

\paragraph{Graph theory.}
All graphs considered in this paper are finite and simple.
Let $G$ be a graph.
For a vertex $v$, the \emph{neighborhood} $N_G(v)$ of $v$ is the set of neighbors of $v$, and the \emph{closed neighborhood} of $v$ is $N_G[v] = N_G(v) \cup \{v\}$.
For a set $X \subseteq V(G)$, the \emph{neighborhood} $N_G(X)$ of $X$ is the set of vertices outside $X$ that have a neighbor in $X$, and the \emph{closed neighborhood} of $X$ is $N_G[X] = X \cup N_G(X)$.
We omit the subscript if the graph is clear from the context.
By $G[X]$ we denote the subgraph of $G$ induced by $X$, and for any set $X$ we write $G - X$ for $G[V(G) \setminus X]$.
A \emph{component} of $G$ is a maximal connected subgraph of $G$, and $\cc(G)$ denotes the number of components of $G$.
A set $X \subseteq V(G)$ is \emph{independent} if no edge of $G$ has both endpoints in $X$, and it is a \emph{vertex cover} of $G$ if every edge of $G$ has an endpoint in $X$.
Thus $X$ is independent if and only if $V(G) \setminus X$ is a vertex cover of $G$.
A vertex cover is \emph{minimal} if no proper subset of it is a vertex cover.
Two disjoint sets $X,Y \subseteq V(G)$ are \emph{anticomplete} if there is no edge between $X$ and $Y$; we also say that $X$ is anticomplete to $Y$.

We denote a path by the sequence $v_0v_1\ldots v_m$ of its vertices; the vertices $v_0$ and $v_m$ are its \emph{endpoints}, and the other vertices are \emph{internal}.
The \emph{length} of a path is the number of its edges.
For a vertex $v$ and a set $A \subseteq V(G)$, the \emph{distance} between $v$ and $A$ in $G$ is the minimum length of a path in $G$ from $v$ to a vertex of $A$, or $\infty$ if there is no such path; in particular, it is $0$ if $v \in A$.
For disjoint sets $A,B \subseteq V(G)$, an \emph{$A$-$B$-path} is a path with one endpoint in $A$, the other endpoint in $B$, and all internal vertices outside $A \cup B$.
By $\dist_G(A,B)$ we denote the minimum length of an $A$-$B$-path in $G$, or $\infty$ if there is no such path.
Note that $\dist_G(A,B)$ is the minimum distance between a vertex of $B$ and $A$.
For pairwise disjoint sets $A,B,Z \subseteq V(G)$, an \emph{$A$-$B$-path in $G$ through $Z$} is an $A$-$B$-path in $G$ all of whose internal vertices belong to $Z$.
The \emph{distance between $A$ and $B$ through $Z$}, denoted by $\dist_G(A,B \mid Z)$, is the minimum length of an $A$-$B$-path in $G$ through $Z$, or $\infty$ if there is no such path.
Equivalently, $\dist_G(A,B \mid Z) = \dist_{G[A \cup B \cup Z]}(A,B)$.
Note that if $Z' \subseteq Z$, then $\dist_G(A,B \mid Z') \geq \dist_G(A,B \mid Z)$.

\paragraph{List 3-Coloring.}
Consider an instance $(G,L)$ of \lcol{3}.
We say that $(G',L')$ is a \emph{subinstance} of $(G,L)$ if $G'$ is an induced subgraph of $G$ and $L'(v) \subseteq L(v)$ for every $v \in V(G')$.
A family $\cI$ of subinstances of $(G,L)$ is \emph{equivalent} to $(G,L)$ if $(G,L)$ is a yes-instance if and only if at least one instance in $\cI$ is a yes-instance.
Note that a subinstance of a subinstance of $(G,L)$ is a subinstance of $(G,L)$.
Consequently, equivalence is transitive: if $\cI$ is equivalent to $(G,L)$, and for every $(G',L') \in \cI$ the family $\cI_{(G',L')}$ of subinstances of $(G',L')$ is equivalent to $(G',L')$, then $\bigcup_{(G',L') \in \cI} \cI_{(G',L')}$ is equivalent to $(G,L)$.

We consider the following two reduction rules for \lcol{3}.
\begin{enumerate}[({R}1)]
    \item If there is a vertex with an empty list, then $(G,L)$ is a no-instance.
    \item If there is a vertex $v$ with $|L(v)| = 1$, then remove the unique color in $L(v)$ from the lists of neighbors of $v$, and remove $v$.
\end{enumerate}
Their correctness is straightforward, and they can be applied in polynomial time.
An instance $(G,L)$ is \emph{reduced} if none of the above rules can be applied to it.

\paragraph{Minimal separators.}
For two disjoint sets $A,B \subseteq V(G)$, a set $S \subseteq V(G) \setminus (A \cup B)$ is an \emph{$A$-$B$-separator} if no component of $G-S$ contains vertices from both $A$ and $B$.
A set $S$ is a \emph{minimal $A$-$B$-separator} if it is an $A$-$B$-separator and no proper subset of $S$ is an $A$-$B$-separator.
For distinct vertices $u,v \in V(G)$, we write \emph{$u$-$v$-separator} instead of $\{u\}$-$\{v\}$-separator.
A \emph{minimal separator} in $G$ is a set that is a minimal $u$-$v$-separator for some $u,v \in V(G)$.
By $\mu(G)$ we denote the number of minimal separators in $G$.
One can readily verify (see e.g.~\cite[p.~3064]{GartlandL23I}) that if $G'$ is an induced subgraph of $G$, then $\mu(G') \leq \mu(G)$.

\begin{theorem}[Berry, Bordat, Cogis~\cite{BerryBC00}]\label{thm:bbc}
Given a graph $G$, all its minimal separators can be enumerated in time $|V(G)|^{\Oh(1)} \cdot (\mu(G)+1)$.
\end{theorem}

\paragraph{Balanced separators.} For $\delta \in (0,1)$, a set $S \subseteq V(G)$ is a \emph{$\delta$-balanced separator} in $G$ if every component of $G-S$ has at most $\delta|V(G)|$ vertices.

\begin{lemma}[Gartland, Lokshtanov~\cite{GartlandL23II}]\label{lem:dbs}
    Let $k \geq 1$ be an integer, and let $G$ be a graph with $\mu(G) < 2^k$.
    Then there is a set $D \subseteq V(G)$ of size at most $2k$ such that $N[D]$ is a $\frac{1}{2}$-balanced separator.
\end{lemma}
We remark that \cref{lem:dbs} follows from combining two results:
that every graph that contains a so-called \emph{$k$-creature} has at least $2^k$ minimal separators~\cite[Lemma~7.6.18]{Gartland23}, and that every graph that does not contain a $k$-creature admits a balanced separator dominated by at most $2k$ vertices~\cite[Lemma~7.4.7]{Gartland23}.

We say that two sets $S_1,S_2 \subseteq V(G)$ are \emph{crossing} in $G$ if there is a component of $G - (S_1 \cap S_2)$ that intersects both $S_1 \setminus S_2$ and $S_2 \setminus S_1$.
Otherwise, $S_1$ and $S_2$ are \emph{non-crossing}.

\begin{theorem}[Gartland, Lokshtanov~\cite{GartlandL23II}]\label{thm:noncrossing}
    Let $G$ be a graph, let $\delta \in (0,1)$, and let $\cS$ be a family of pairwise non-crossing $\delta$-balanced separators in $G$ such that no vertex of $G$ belongs to all sets in $\cS$.
    Then every component of $G$ has at most $\delta|V(G)|$ vertices.
\end{theorem}

\section{Separating two sets}

In this section we prove \cref{thm:cut}, which is one of the key ingredients of our algorithm.
The setting is as follows.
We are given a graph $G$ with fewer than $2^k$ minimal separators and an instance $(G',L')$ of \lcol{3}, where $G'$ is an induced subgraph of $G$.
We are also given pairwise disjoint sets $A,B,Z \subseteq V(G)$, where $G[A \cup B]$ has at most $c$ components and $Z \subseteq V(G')$.
Note that $A$ and $B$ do not need to be contained in $V(G')$; this is why, apart from the instance $(G',L')$, we need the graph $G$.
Our aim is to destroy all $A$-$B$-paths in $G$ through $Z$ by removing vertices of $Z$ from the instance, i.e., to obtain subinstances $(G'',L'')$ such that $\dist_G(A,B \mid Z \cap V(G'')) = \infty$.
We show that this can be done if $\dist_G(A,B \mid Z) \geq 3$, i.e., $A$ and $B$ are anticomplete and do not have common neighbors in $Z$.

\begin{theorem}\label{thm:cut}
    Let $G$ be an $n$-vertex graph with $\mu(G) < 2^k$, and let $(G',L')$ be an instance of \lcol{3}, where $G'$ is an induced subgraph of $G$.
    Let $A,B,Z \subseteq V(G)$ be pairwise disjoint sets such that $Z \subseteq V(G')$ and $G[A \cup B]$ has at most $c$ components.
    Suppose that $\dist_G(A,B \mid Z) \geq 3$.
    Then in time $n^{\Oh(1)} \cdot (4n)^{6kc}$ one can compute a family $\cI$ of at most $(2n)^{6kc}$ subinstances of $(G',L')$ such that:
    \begin{enumerate}
        \item\label{it:cut:equiv} $\cI$ is equivalent to $(G',L')$, and
        \item\label{it:cut:dist} for every $(G'',L'') \in \cI$, we have $\dist_G(A,B \mid Z \cap V(G'')) = \infty$.
    \end{enumerate}
\end{theorem}

We prove \cref{thm:cut} by gradually increasing the distance between $A$ and $B$ through the surviving vertices of $Z$, i.e., the value $\dist_G(A,B \mid Z \cap V(G''))$ for the constructed subinstances $(G'',L'')$: first from $3$ to $4$, then to $5$, and finally to $\infty$.

\subsection{From distance 3 to 4}\label{sec:dist4}

We start with two purely graph-theoretic lemmas.

\begin{lemma}\label{lem:union}
    Let $G$ be a graph and let $A,B \subseteq V(G)$ be disjoint nonempty sets such that $G[A \cup B]$ has at most $c$ components.
    Then every minimal $A$-$B$-separator in $G$ is the union of at most $c-1$ minimal separators in $G$.
\end{lemma}
\begin{proof}
    Let $S$ be a minimal $A$-$B$-separator.
    We define a sequence $S_0 \subseteq S_1 \subseteq \ldots$ of subsets of $S$ as follows.
    Let $S_0 = \emptyset$, and for $i \geq 0$ proceed as follows.
    If $S_i$ is an $A$-$B$-separator, stop.
    Otherwise, choose $u_i \in A$ and $v_i \in B$ that belong to the same component of $G - S_i$.
    Since $S$ is an $A$-$B$-separator, it is also a $u_i$-$v_i$-separator, so it contains a minimal $u_i$-$v_i$-separator $P_i$, which is a minimal separator in $G$.
    Set $S_{i+1} = S_i \cup P_i$.

    For $T \subseteq S$, let $q(T)$ be the number of components of $G-T$ that intersect $A \cup B$.
    Since $T \cap (A \cup B) = \emptyset$, each component of $G[A \cup B]$ is contained in exactly one component of $G - T$, and each component of $G - T$ that intersects $A \cup B$ contains a component of $G[A \cup B]$.
    Thus $1 \leq q(T) \leq c$.

    \begin{claim}\label{cl:union:q}
        If $S_{i+1}$ is defined, then $q(S_{i+1}) > q(S_i)$.
    \end{claim}
    \begin{claimproof}
        As $S_i \subseteq S_{i+1}$, every component of $G - S_{i+1}$ is contained in a component of $G - S_i$.
        Every component of $G-S_i$ that intersects $A \cup B$ contains at least one component of $G-S_{i+1}$ that intersects $A \cup B$.
        Moreover, the component of $G - S_i$ containing $u_i$ and $v_i$ contains at least two such components, as $u_i$ and $v_i$ are separated by $P_i \subseteq S_{i+1}$.
    \end{claimproof}

    By \cref{cl:union:q}, the process stops at some $S_j$ with $j \leq c-1$.
    The set $S_j \subseteq S$ is an $A$-$B$-separator and the union of $j$ minimal separators in $G$.
    By the minimality of $S$, we have $S_j = S$.
\end{proof}

Let $G$ be a graph and let $A,B \subseteq V(G)$ be disjoint sets with $\dist_G(A,B) \geq 3$, i.e., $N[A] \cap N[B] = \emptyset$.
The \emph{interface graph} of $A$ and $B$ in $G$ is the bipartite graph with vertex set $N(A) \cup N(B)$ whose edges are all edges of $G$ between $N(A)$ and $N(B)$.
Note that the $A$-$B$-paths of length $3$ in $G$ are exactly the paths $axyb$ such that $a \in A$, $b \in B$, and $xy$ is an edge of the interface graph.

\begin{lemma}\label{lem:interface}
    Let $G$ be a graph with $\mu(G) < 2^k$, and let $A,B \subseteq V(G)$ be disjoint nonempty sets with $\dist_G(A,B) \geq 3$, such that $G[A \cup B]$ has at most $c$ components.
    Let $H$ be the interface graph of $A$ and $B$ in $G$.
    Then every minimal vertex cover of $H$ is the union of at most $c-1$ minimal separators in $G[N[A \cup B]]$.
    Consequently, $H$ has at most $2^{k(c-1)}$ minimal vertex covers, and they can be enumerated in time $|V(G)|^{\Oh(1)} \cdot 2^{k(c-1)}$.
\end{lemma}
\begin{proof}
    Replacing $G$ by $G[N[A \cup B]]$ changes neither $H$ nor $G[A \cup B]$, and does not increase the number of minimal separators.
    Thus we may assume that $V(G) = N[A \cup B]$.
    Let $X = N(A)$ and $Y = N(B)$.
    Since $N[A] \cap N[B] = \emptyset$, the sets $A,B,X,Y$ are pairwise disjoint, $A$ is anticomplete to $B \cup Y$, and $B$ is anticomplete to $X$.
    Thus the only edges of $G$ between $A \cup X$ and $B \cup Y$ are the edges between $X$ and $Y$, i.e., the edges of $H$.

    \begin{claim}\label{cl:interface:vc}
        The $A$-$B$-separators in $G$ are exactly the vertex covers of $H$.
    \end{claim}
    \begin{claimproof}
        Both the $A$-$B$-separators in $G$ and the vertex covers of $H$ are subsets of $V(G) \setminus (A \cup B) = V(H)$.
        If $C$ is a vertex cover of $H$, then $G - C$ has no edges between $(A \cup X) \setminus C$ and $(B \cup Y) \setminus C$, so $C$ is an $A$-$B$-separator in $G$.
        Conversely, if some edge $xy$ of $H$ with $x \in X$ and $y \in Y$ satisfies $x,y \notin C$, then for any neighbor $a \in A$ of $x$ and any neighbor $b \in B$ of $y$, the path $axyb$ is contained in $G - C$, so $C$ is not an $A$-$B$-separator in $G$.
    \end{claimproof}

    By \cref{cl:interface:vc}, the minimal vertex covers of $H$ are exactly the minimal $A$-$B$-separators in $G$.
    By \cref{lem:union}, each of them is the union of at most $c-1$ minimal separators in $G$.
    This proves the first statement.
    Thus each minimal vertex cover of $H$ is the union of a sequence of exactly $c-1$ sets, each of which is $\emptyset$ or a minimal separator in $G$.
    As there are at most $\mu(G)+1 \leq 2^k$ choices for each set, $H$ has at most $2^{k(c-1)}$ minimal vertex covers.
    To enumerate them, we list all minimal separators in $G$ using \cref{thm:bbc} and check all such unions; this takes time $|V(G)|^{\Oh(1)} \cdot 2^{k(c-1)}$, as $c \geq 2$ because $A$ and $B$ are nonempty and anticomplete.
\end{proof}

Now we are ready to increase the distance from $3$ to $4$.

\begin{lemma}\label{lem:dist4}
    Let $G$ be an $n$-vertex graph with $\mu(G) < 2^k$, and let $(G',L')$ be an instance of \lcol{3}, where $G'$ is an induced subgraph of $G$.
    Let $A,B,Z \subseteq V(G)$ be pairwise disjoint sets such that $A$ and $B$ are nonempty, $Z \subseteq V(G')$, and $G[A \cup B]$ has at most $c$ components.
    Suppose that $\dist_G(A,B \mid Z) \geq 3$.
    Then in time $n^{\Oh(1)} \cdot 2^{3kc}$ one can compute a family $\cI$ of at most $2^{3kc}$ subinstances of $(G',L')$ such that:
    \begin{enumerate}
        \item $\cI$ is equivalent to $(G',L')$, and
        \item for every $(G'',L'') \in \cI$, we have $\dist_G(A,B \mid Z \cap V(G'')) \geq 4$.
    \end{enumerate}
\end{lemma}
\begin{proof}
    Let $F = G[A \cup B \cup Z]$, and recall that $\dist_F(A,B) = \dist_G(A,B \mid Z) \geq 3$.
    Let $H$ be the interface graph of $A$ and $B$ in $F$.
    Its vertex set $N_F(A) \cup N_F(B)$ is contained in $Z \subseteq V(G')$, and since $G'$ is an induced subgraph of $G$, every edge of $H$ is an edge of $G'$.
    Moreover, $\mu(F) \leq \mu(G) < 2^k$ and $F[A \cup B] = G[A \cup B]$ has at most $c$ components.
    By \cref{lem:interface}, the family $\cC$ of all minimal vertex covers of $H$ has size at most $2^{k(c-1)}$ and can be computed in time $n^{\Oh(1)} \cdot 2^{k(c-1)}$.

    For every triple $\mathbf{C} = (C_1,C_2,C_3) \in \cC^3$, we define the list assignment $L_{\mathbf{C}}$ of $G'$ by $L_{\mathbf{C}}(v) = L'(v) \setminus \{i \mid v \in C_i\}$ for $v \in V(G')$.
    We exhaustively apply the reduction rules to $(G',L_{\mathbf{C}})$; if (R1) is applied, we discard the instance, and otherwise we include the obtained reduced instance in $\cI$.
    We have $|\cI| \leq |\cC|^3 \leq 2^{3k(c-1)} \leq 2^{3kc}$, and the running time is $n^{\Oh(1)} \cdot 2^{3kc}$.

    Let us show that $\cI$ is equivalent to $(G',L')$.
    As the reduction rules are correct and $L_{\mathbf{C}}(v) \subseteq L'(v)$ for every $v \in V(G')$, if some instance in $\cI$ is a yes-instance, then so is $(G',L')$.
    Conversely, let $\psi$ be a proper coloring of $(G',L')$.
    For each $i \in [3]$, the set $\psi^{-1}(i) \cap V(H)$ is independent in $G'$, and thus in $H$, so $V(H) \setminus \psi^{-1}(i)$ is a vertex cover of $H$; let $C_i \in \cC$ be a minimal vertex cover of $H$ contained in it, and let $\mathbf{C} = (C_1,C_2,C_3)$.
    For every $v \in V(G')$, we have $v \notin C_{\psi(v)}$, so $\psi(v) \in L_{\mathbf{C}}(v)$.
    Hence, $\psi$ is a proper coloring of $(G',L_{\mathbf{C}})$, so by the correctness of the reduction rules, the instance obtained from $\mathbf{C}$ is not discarded and is a yes-instance.

    \begin{claim}\label{cl:dist4:edgeless}
        For every $(G'',L'') \in \cI$, no edge of $H$ has both endpoints in $V(G'')$.
    \end{claim}
    \begin{claimproof}
        Let $(G'',L'') \in \cI$ be obtained from $\mathbf{C}$, and suppose that $xy$ is an edge of $H$ with $x,y \in V(G'')$.
        Since $(G'',L'')$ is reduced, we have $|L''(x)|,|L''(y)| \geq 2$.
        On the other hand, $L''(x) \subseteq L_{\mathbf{C}}(x)$ and $L''(y) \subseteq L_{\mathbf{C}}(y)$, and these sets are disjoint: for every color $i$, the vertex cover $C_i$ contains $x$ or $y$, so $i$ is removed from $L'(x)$ or from $L'(y)$.
        Two disjoint subsets of $[3]$ cannot both have at least two elements, a contradiction.
    \end{claimproof}

    Finally, consider $(G'',L'') \in \cI$.
    As $Z \cap V(G'') \subseteq Z$, we have $\dist_G(A,B \mid Z \cap V(G'')) \geq \dist_G(A,B \mid Z) \geq 3$.
    Every $A$-$B$-path in $G$ through $Z \cap V(G'')$ is an $A$-$B$-path in $F$, so if it has length $3$, then it is of the form $axyb$, where $xy$ is an edge of $H$ and $x,y \in V(G'')$.
    By \cref{cl:dist4:edgeless}, there is no such path, and thus $\dist_G(A,B \mid Z \cap V(G'')) \geq 4$.
\end{proof}

\subsection{From distance 4 to 5}\label{sec:dist5}

To increase the distance from $4$ to $5$, we move all vertices of $Z$ adjacent to $B$ into $B$, and apply \cref{lem:dist4}.

\begin{lemma}\label{lem:dist5}
    Let $G$ be an $n$-vertex graph with $\mu(G) < 2^k$, and let $(G',L')$ be an instance of \lcol{3}, where $G'$ is an induced subgraph of $G$.
    Let $A,B,Z \subseteq V(G)$ be pairwise disjoint sets such that $A$ and $B$ are nonempty, $Z \subseteq V(G')$, and $G[A \cup B]$ has at most $c$ components.
    Suppose that $\dist_G(A,B \mid Z) \geq 4$.
    Then in time $n^{\Oh(1)} \cdot 2^{3kc}$ one can compute a family $\cI$ of at most $2^{3kc}$ subinstances of $(G',L')$ such that:
    \begin{enumerate}
        \item $\cI$ is equivalent to $(G',L')$, and
        \item for every $(G'',L'') \in \cI$, we have $\dist_G(A,B \mid Z \cap V(G'')) \geq 5$.
    \end{enumerate}
\end{lemma}
\begin{proof}
    Let $F = G[A \cup B \cup Z]$, and recall that $\dist_F(A,B) = \dist_G(A,B \mid Z) \geq 4$.
    Let $B^+ = N_F[B]$ and $Z^- = Z \setminus B^+$.
    As every vertex of $B^+ \setminus B$ has a neighbor in $B$, every vertex of $B^+$ is at distance at least $3$ from $A$ in $F$.
    In particular, $B^+ \subseteq B \cup Z$, so the sets $A,B^+,Z^-$ are pairwise disjoint; moreover, $B^+ \neq \emptyset$ and $Z^- \subseteq V(G')$.
    Every component of $G[A \cup B^+]$ intersects $A \cup B$, so $G[A \cup B^+]$ has at most $\cc(G[A \cup B]) \leq c$ components.
    Finally, every $A$-$B^+$-path in $G$ through $Z^-$ is a path in $F$ that ends in $B^+$, so it has length at least $3$, i.e., $\dist_G(A,B^+ \mid Z^-) \geq 3$.

    Thus we can apply \cref{lem:dist4} to $(G',L')$ and $A,B^+,Z^-$.
    In time $n^{\Oh(1)} \cdot 2^{3kc}$ we obtain a family $\cI$ of at most $2^{3kc}$ subinstances of $(G',L')$, which is equivalent to $(G',L')$, and such that $\dist_G(A,B^+ \mid Z^- \cap V(G'')) \geq 4$ for every $(G'',L'') \in \cI$.

    Consider $(G'',L'') \in \cI$.
    As $Z \cap V(G'') \subseteq Z$, we have $\dist_G(A,B \mid Z \cap V(G'')) \geq \dist_G(A,B \mid Z) \geq 4$.
    Suppose that there is an $A$-$B$-path in $G$ through $Z \cap V(G'')$ of length exactly $4$, say $az_1z_2z_3b$, where $a \in A$, $z_1,z_2,z_3 \in Z \cap V(G'')$, and $b \in B$.
    Then $z_3 \in B^+$, and $z_1,z_2 \notin B^+$, as they are at distance at most $2$ from $A$ in $F$.
    Thus $az_1z_2z_3$ is an $A$-$B^+$-path in $G$ through $Z^- \cap V(G'')$ of length $3$, a contradiction.
\end{proof}

Combining \cref{lem:dist4,lem:dist5}, we obtain the following.

\begin{corollary}\label{cor:dist5}
    Let $G$ be an $n$-vertex graph with $\mu(G) < 2^k$, and let $(G',L')$ be an instance of \lcol{3}, where $G'$ is an induced subgraph of $G$.
    Let $A,B,Z \subseteq V(G)$ be pairwise disjoint sets such that $A$ and $B$ are nonempty, $Z \subseteq V(G')$, and $G[A \cup B]$ has at most $c$ components.
    Suppose that $\dist_G(A,B \mid Z) \geq 3$.
    Then in time $n^{\Oh(1)} \cdot 2^{6kc}$ one can compute a family $\cI$ of at most $2^{6kc}$ subinstances of $(G',L')$ such that:
    \begin{enumerate}
        \item $\cI$ is equivalent to $(G',L')$, and
        \item for every $(G'',L'') \in \cI$, we have $\dist_G(A,B \mid Z \cap V(G'')) \geq 5$.
    \end{enumerate}
\end{corollary}
\begin{proof}
    We apply \cref{lem:dist4} to $(G',L')$ and $A,B,Z$.
    For every obtained instance $(G_1,L_1)$, we have $\dist_G(A,B \mid Z \cap V(G_1)) \geq 4$, so we can apply \cref{lem:dist5} to $(G_1,L_1)$ and $A,B,Z \cap V(G_1)$.
    Let $\cI$ be the family of all instances obtained in this way.
    By the transitivity of equivalence, $\cI$ is a family of at most $2^{3kc} \cdot 2^{3kc}$ subinstances of $(G',L')$ that is equivalent to $(G',L')$, and it is computed in time $n^{\Oh(1)} \cdot 2^{3kc} \cdot 2^{3kc}$.
    Moreover, if $(G'',L'') \in \cI$ is obtained from $(G_1,L_1)$, then $Z \cap V(G'') = (Z \cap V(G_1)) \cap V(G'')$, and thus $\dist_G(A,B \mid Z \cap V(G'')) \geq 5$.
\end{proof}

\subsection{From distance 5 to \texorpdfstring{$\infty$}{infinity}}\label{sec:distinf}

The crucial property of distance $5$ is that no vertex is at distance at most $2$ from both $A$ and $B$.
The following lemma uses this to reduce the problem to one where the distance is again at least $3$, but the set $Z$ is at most half as large.

\begin{lemma}\label{lem:halving}
    Let $G$ be a graph and let $A,B,Z \subseteq V(G)$ be pairwise disjoint sets such that $\dist_G(A,B \mid Z) \geq 5$.
    Then in polynomial time one can compute sets $A',B',Z' \subseteq V(G)$ such that:
    \begin{enumerate}
        \item\label{it:halving:sets} $A',B',Z'$ are pairwise disjoint, $A \subseteq A'$, $B \subseteq B'$, $Z' \subseteq Z$, and $|Z'| \leq |Z|/2$,
        \item\label{it:halving:cc} $\cc(G[A' \cup B']) \leq \cc(G[A \cup B])$,
        \item\label{it:halving:dist} $\dist_G(A',B' \mid Z') \geq 3$, and
        \item\label{it:halving:paths} every $A$-$B$-path in $G$ through $Z$ contains a subpath that is an $A'$-$B'$-path in $G$ through $Z'$.
    \end{enumerate}
\end{lemma}
\begin{proof}
    Let $F = G[A \cup B \cup Z]$, and recall that $\dist_F(A,B) = \dist_G(A,B \mid Z) \geq 5$.
    Let $Q = N_F[N_F[B]]$ be the set of vertices at distance at most $2$ from $B$ in $F$.
    Every vertex of $Q$ is at distance at least $3$ from $A$ in $F$; in particular, $A \cap Q = \emptyset$.
    Let $C$ be the union of the vertex sets of all components of $F - Q$ that intersect $A$; intuitively, $C$ consists of the vertices that can be reached from $A$ without coming close to $B$.
    Then $A \subseteq C$, and every vertex of $C$ is at distance at least $3$ from $B$ in $F$, as $C \cap Q = \emptyset$.
    As $B \subseteq Q$, we have $C \subseteq A \cup Z$, so $Z$ is the disjoint union of $C \setminus A$ and $Z \setminus C$.
    If $|C \setminus A| \leq |Z|/2$, we set $(A',B',Z') = (A,Q,C \setminus A)$, and otherwise we set $(A',B',Z') = (C,B,Z \setminus C)$.
    In both cases, $|Z'| \leq |Z|/2$, and clearly these sets can be computed in polynomial time.

    We first verify conditions~\ref{it:halving:cc} and~\ref{it:halving:dist}, simultaneously for both cases.
    Every vertex of $Q$ is joined to $B$ by a path all of whose vertices belong to $Q$, and every vertex of $C$ is joined to $A$ by a path all of whose vertices belong to $C$.
    Hence every component of $G[A' \cup B']$ intersects $A \cup B \subseteq A' \cup B'$, and thus $\cc(G[A' \cup B']) \leq \cc(G[A \cup B])$.
    Moreover, $A' \cup B' \cup Z' \subseteq V(F)$, so every $A'$-$B'$-path through $Z'$ is a path in $F$ that joins a vertex of $A'$ with a vertex of $B'$.
    In the first case, its endpoint in $B' = Q$ is at distance at least $3$ from $A' = A$ in $F$, and in the second case, its endpoint in $A' = C$ is at distance at least $3$ from $B' = B$ in $F$.
    Thus $\dist_G(A',B' \mid Z') \geq 3$.

    Now suppose that $(A',B',Z') = (A,Q,C \setminus A)$.
    Condition~\ref{it:halving:sets} holds as $A \cap Q = \emptyset$, $C \cap Q = \emptyset$, $B \subseteq Q$, and $C \setminus A \subseteq Z$.
    For condition~\ref{it:halving:paths}, let $v_0v_1\ldots v_m$ be an $A$-$B$-path through $Z$, where $v_0 \in A$ and $v_m \in B$.
    As $v_0 \notin Q$ and $v_m \in Q$, there is a minimum index $i \geq 1$ such that $v_i \in Q$.
    Then $v_0v_1\ldots v_{i-1}$ is a path in $F - Q$ that contains $v_0 \in A$, so $v_0,\ldots,v_{i-1} \in C$.
    As $v_1,\ldots,v_{i-1} \in Z$, they belong to $C \setminus A$, and thus $v_0v_1\ldots v_i$ is an $A$-$Q$-path through $C \setminus A$.

    Finally, suppose that $(A',B',Z') = (C,B,Z \setminus C)$.
    Condition~\ref{it:halving:sets} holds as $A \subseteq C$ and $C \cap B \subseteq C \cap Q = \emptyset$.
    For condition~\ref{it:halving:paths}, let $v_0v_1\ldots v_m$ be an $A$-$B$-path through $Z$, where $v_0 \in A$ and $v_m \in B$.
    As $v_0 \in C$ and $v_m \notin C$, there is a maximum index $i \leq m-1$ such that $v_i \in C$.
    Then $v_{i+1},\ldots,v_{m-1} \in Z \setminus C$, and thus $v_iv_{i+1}\ldots v_m$ is a $C$-$B$-path through $Z \setminus C$.
\end{proof}

Now we are ready to prove \cref{thm:cut}.
We use \cref{cor:dist5} to increase the distance between $A$ and $B$ to $5$, and then \cref{lem:halving} to reduce the problem to one where $Z$ is at most half as large, which we solve recursively.
As $Z$ is halved in each step, the recursion terminates after $\Oh(\log n)$ steps.

\begin{proof}[Proof of \cref{thm:cut}]
    We describe a recursive algorithm, and prove its correctness by induction on $|Z|$.
    If $\dist_G(A,B \mid Z) = \infty$, then the algorithm returns $\cI = \{(G',L')\}$, which clearly satisfies conditions~\ref{it:cut:equiv} and~\ref{it:cut:dist}.
    Otherwise, $A$ and $B$ are nonempty, and $Z \neq \emptyset$, as $A$ and $B$ are anticomplete.
    By \cref{cor:dist5}, in time $n^{\Oh(1)} \cdot 2^{6kc}$ we compute a family $\cJ$ of at most $2^{6kc}$ subinstances of $(G',L')$, which is equivalent to $(G',L')$, such that $\dist_G(A,B \mid Z \cap V(G_2)) \geq 5$ for every $(G_2,L_2) \in \cJ$.

    Consider $(G_2,L_2) \in \cJ$.
    We apply \cref{lem:halving} to $A,B,Z \cap V(G_2)$, and obtain sets $A',B',Z'$.
    By \cref{lem:halving}, these sets are pairwise disjoint, $Z' \subseteq V(G_2)$, $G[A' \cup B']$ has at most $\cc(G[A \cup B]) \leq c$ components, $\dist_G(A',B' \mid Z') \geq 3$, and $|Z'| \leq |Z|/2 < |Z|$.
    As $G_2$ is an induced subgraph of $G$, we can call the algorithm recursively for $(G_2,L_2)$ and $A',B',Z'$.
    By the induction hypothesis, it returns a family $\cI_{(G_2,L_2)}$ of subinstances of $(G_2,L_2)$, which is equivalent to $(G_2,L_2)$, such that $\dist_G(A',B' \mid Z' \cap V(G'')) = \infty$ for every $(G'',L'') \in \cI_{(G_2,L_2)}$.
    The algorithm returns $\cI = \bigcup_{(G_2,L_2) \in \cJ} \cI_{(G_2,L_2)}$.

    By the transitivity of equivalence, $\cI$ is a family of subinstances of $(G',L')$ that is equivalent to $(G',L')$.
    Now consider $(G'',L'') \in \cI_{(G_2,L_2)}$ for some $(G_2,L_2) \in \cJ$, and suppose that there is an $A$-$B$-path in $G$ through $Z \cap V(G'')$.
    As $V(G'') \subseteq V(G_2)$, this path is an $A$-$B$-path through $Z \cap V(G_2)$.
    By \cref{lem:halving}, it contains a subpath that is an $A'$-$B'$-path through $Z'$.
    The internal vertices of this subpath are internal vertices of the original path, so they belong to $Z' \cap V(G'')$.
    Thus the subpath is an $A'$-$B'$-path in $G$ through $Z' \cap V(G'')$, a contradiction.

    \medskip
    \noindent\textbf{Size and running time.}
    Every call makes at most $2^{6kc}$ recursive calls, in each of which $Z$ is at most half as large, and only calls with $Z \neq \emptyset$ make recursive calls.
    As $|Z| \leq n$, the recursion tree has depth at most $\log n + 1$, so it has at most $2^{6kc(\log n + 1)} = (2n)^{6kc}$ leaves and at most $2 \cdot (2n)^{6kc}$ nodes.
    Every instance in $\cI$ originates from a leaf, and each leaf returns at most one instance, so $|\cI| \leq (2n)^{6kc}$.
    Apart from the recursive calls, each call takes time $n^{\Oh(1)} \cdot 2^{6kc}$ by \cref{cor:dist5} and \cref{lem:halving}.
    Thus the total running time is $n^{\Oh(1)} \cdot 2^{6kc} \cdot (2n)^{6kc} = n^{\Oh(1)} \cdot (4n)^{6kc}$.
\end{proof}

\section{Proof of the main theorem}

In this section we prove the main result of the paper.

\begin{theorem}\label{thm:realmain}
    Every instance $(G,L)$ of \lcol{3}, where $G$ has $n \geq 2$ vertices and $\mu(G) < 2^k$ for an integer $k \geq 1$, can be solved in time $n^{\Oh(k^2 \log^3 n)}$.
\end{theorem}

Observe that \cref{thm:realmain} implies \cref{thm:main}.
Indeed, if $\cG$ is a tame (resp., quasitame) class of graphs, then every $n$-vertex graph $G \in \cG$ with $n \geq 2$ satisfies $\mu(G) < 2^k$ for some $k = \Oh(\log n)$ (resp., $k = \Oh(\log^a n)$, where $a$ is the constant from the definition of quasitame classes), so by \cref{thm:realmain}, \lcol{3} can be solved on $\cG$ in time $n^{\Oh(\log^5 n)}$ (resp., $n^{\Oh(\log^{2a+3} n)}$), which is quasipolynomial.

The main technical step is the following lemma, which encapsulates the branching scheme.

\begin{lemma}\label{lem:branching}
    Let $(G,L)$ be an instance of \lcol{3}, where $G$ has $n \geq 2$ vertices and $\mu(G) < 2^k$ for an integer $k \geq 1$.
    Then in time $n^{\Oh(k^2 \log^2 n)}$ one can compute a family $\cI$ of subinstances of $(G,L)$ such that:
    \begin{enumerate}
        \item\label{it:branching:equiv} $\cI$ is equivalent to $(G,L)$, and
        \item\label{it:branching:small} for every $(G',L') \in \cI$, every component of $G'$ has at most $n/2$ vertices.
    \end{enumerate}
\end{lemma}

Let us first show how \cref{lem:branching} implies \cref{thm:realmain}.

\begin{proof}[Proof of \cref{thm:realmain} assuming \cref{lem:branching}]
    We apply \cref{lem:branching} to $(G,L)$, and for every obtained instance, we recursively solve its restrictions to all components of its graph; instances with one vertex are solved directly.
    Correctness follows from \cref{lem:branching}.
    The graphs of all instances in the recursion are induced subgraphs of $G$, so they have fewer than $2^k$ minimal separators, and the number of vertices halves in each recursive call.
    Thus the recursion has depth $\Oh(\log n)$, and each call takes time $n^{\Oh(k^2 \log^2 n)}$ and makes at most $n^{\Oh(k^2 \log^2 n)}$ recursive calls, which gives the running time $n^{\Oh(k^2 \log^3 n)}$.
\end{proof}

The rest of this section is devoted to the proof of \cref{lem:branching}.
We will repeatedly use the following operation.
Let $(G',L')$ be an instance of \lcol{3}, and let $X \subseteq V(G')$ be nonempty.
\emph{Guessing the coloring of $X$} works as follows.
For every function $\phi$ that assigns to each $x \in X$ a color $\phi(x) \in L'(x)$, we restrict the list of every $x \in X$ to $\{\phi(x)\}$, and exhaustively apply the reduction rules.
We discard every instance to which (R1) is applied, and return the family of all remaining reduced instances.
This family consists of at most $3^{|X|}$ subinstances of $(G',L')$, and it is equivalent to $(G',L')$, as we consider all colorings of $X$ and the reduction rules are correct.
For a single vertex $v$, we write \emph{guessing the color of $v$} instead of guessing the coloring of $\{v\}$.

\begin{proof}[Proof of \cref{lem:branching}]
    Let us start with an overview.
    We repeatedly find a set $D$ of at most $2k$ vertices whose closed neighborhood is a balanced separator, and guess the coloring of $D$; we call $D$ a \emph{center}.
    Afterwards, every remaining neighbor of $D$ has only two colors in its list.
    Once we have $\Oh(\log n)$ pairwise anticomplete centers such that no vertex is adjacent to all of them, their neighborhoods are balanced separators with empty intersection.
    We make them pairwise non-crossing using \cref{thm:cut}, and conclude using \cref{thm:noncrossing}.
    To ensure that no vertex is adjacent to too many centers, we guess the colors of vertices whose closed neighborhoods contain many vertices adjacent to centers.
    In all branches but one, many of them are deleted, which is captured by a potential function.

    \medskip
    \noindent\textbf{The procedure.}
    First, we exhaustively apply the reduction rules to $(G,L)$; if (R1) is applied, we output the empty family, which is equivalent to $(G,L)$.
    Thus we may assume that $(G,L)$ is reduced.
    Let $h$ be the smallest integer such that $2^h > n$; note that $h = \Oh(\log n)$.

    The procedure maintains a reduced subinstance $(G',L')$ of $(G,L)$, initially equal to $(G,L)$, and a sequence $D_1,\ldots,D_t$ of \emph{centers}, initially empty, which are pairwise disjoint subsets of $V(G) \setminus V(G')$.
    For $v \in V(G')$, let $\m(v)$ be the number of centers that contain a neighbor of $v$, and let $\w(v) = 2^{\m(v)}-1$.
    For $X \subseteq V(G')$, let $\w(X) = \sum_{v \in X} \w(v)$; we call $\w(V(G'))$ the \emph{potential} of $(G',L')$.
    A vertex $v \in V(G')$ is \emph{heavy} if $\w(N_{G'}[v]) \geq n$.
    The \emph{greedy selection} repeatedly selects the remaining center with the largest index, and removes it together with all remaining centers that are not anticomplete to it, until no center remains; clearly, the selected centers are pairwise anticomplete.
    In each instance, the procedure applies the first applicable step among the following ones.
    \begin{enumerate}[({S}1)]
        \item If every component of $G'$ has at most $n/2$ vertices, output $(G',L')$.
        \item If some vertex $v \in V(G')$ is heavy, guess the color of $v$, and continue with each obtained instance.
        \item If the greedy selection selects fewer than $h$ centers, find a smallest set $D \subseteq V(G')$ such that $N_{G'}[D]$ is a $\frac{1}{2}$-balanced separator in $G'$, by checking all sets in the order of increasing size.
        Note that $D \neq \emptyset$, as (S1) does not apply.
        Guess the coloring of $D$, and continue with each obtained instance, with $D_{t+1} = D$ as a new center.
        \item Otherwise, apply the finishing step described below to $(G',L')$ and the first $h$ centers selected by the greedy selection, and output the obtained family.
    \end{enumerate}
    Note that the procedure does not use $k$, and neither does the algorithm of \cref{thm:cut}, which we apply in the finishing step.

    \medskip
    \noindent\textbf{Analysis of steps (S1)--(S3).}
    As $\mu(G') \leq \mu(G) < 2^k$, by \cref{lem:dbs}, every center has at most $2k$ vertices.
    Let $M = 2kh^2$.

    \begin{claim}\label{cl:br:lists}
        Every vertex $v \in V(G')$ with $\m(v) \geq 1$ satisfies $|L'(v)| = 2$.
    \end{claim}
    \begin{claimproof}
        Let $x$ be a neighbor of $v$ in some center.
        When we guessed the coloring of this center, the vertex $x$ was deleted by (R2), and thus its color was removed from the list of $v$.
        As lists never grow, we have $|L'(v)| \leq 2$, and $|L'(v)| \geq 2$, as $(G',L')$ is reduced.
    \end{claimproof}

    \begin{claim}\label{cl:br:heavy}
        Let $v \in V(G')$ be heavy.
        For every color $i \in L'(v)$ except at most one, the instance obtained for $i$ when guessing the color of $v$, if not discarded, has potential at most $\w(V(G')) - n/2$.
    \end{claim}
    \begin{claimproof}
        For $i \in [3]$, let $X_i = \{u \in N_{G'}[v] \mid i \in L'(u)\}$.
        By \cref{cl:br:lists}, every vertex $u$ with $\w(u) > 0$ has exactly two colors in its list, so $\w(X_1) + \w(X_2) + \w(X_3) = 2\w(N_{G'}[v])$.
        As $\w(X_i) \leq \w(N_{G'}[v])$ for every $i$, at most one color $i$ satisfies $\w(X_i) < \w(N_{G'}[v])/2$.
        Consider any other color $i \in L'(v)$, so that $\w(X_i) \geq \w(N_{G'}[v])/2 \geq n/2$, and suppose that the instance $(G'',L'')$ obtained for $i$ is not discarded.
        Every vertex $u \in X_i$ with $\w(u) > 0$ is deleted in $(G'',L'')$: either $u = v$, or $u$ is a neighbor of $v$, so it loses the color $i$ and is left with at most one color, while $(G'',L'')$ is reduced.
        As the weights of the remaining vertices do not change, the potential decreases by at least $\w(X_i) \geq n/2$.
    \end{claimproof}

    \begin{claim}\label{cl:br:mult}
        If no vertex of $G'$ is heavy, then $\m(v) < h$ for every $v \in V(G')$.
    \end{claim}
    \begin{claimproof}
        We have $2^{\m(v)} - 1 = \w(v) \leq \w(N_{G'}[v]) < n$, so $2^{\m(v)} \leq n < 2^h$.
    \end{claimproof}

    \begin{claim}\label{cl:br:greedy}
        If $t \geq M$, then the greedy selection selects at least $h$ centers.
        Consequently, we always have $t \leq M$.
    \end{claim}
    \begin{claimproof}
        Consider $s \in [t]$.
        When $D_s$ was found in (S3), it was contained in the vertex set of the current instance, and no vertex of this instance was heavy.
        Thus, by \cref{cl:br:mult}, every vertex of $D_s$ has a neighbor in fewer than $h$ of the centers $D_1,\ldots,D_{s-1}$, so $D_s$ is anticomplete to all but at most $2k(h-1)$ of them.
        Hence, each round of the greedy selection removes at most $2k(h-1)+1 \leq 2kh$ centers, so there are at least $t/(2kh)$ rounds, which is at least $h$ if $t \geq M$.
        In particular, (S3) is not applied if $t \geq M$, so we always have $t \leq M$.
    \end{claimproof}

    \begin{claim}\label{cl:br:increase}
        Every execution of (S3) increases the potential by less than $(2k+1)n$.
    \end{claim}
    \begin{claimproof}
        Let $D$ be the center found in (S3) for $(G',L')$, and let $(G'',L'')$ be an obtained instance.
        Compared to $(G',L')$, some vertices are deleted, and for every $v \in V(G'')$ that has a neighbor in $D$, the weight $\w(v)$ changes to $2\w(v)+1$; all such vertices $v$ belong to $N_{G'}[D]$.
        As no vertex of $G'$ is heavy, the potential increases by at most
        \[
            \sum_{v \in N_{G'}[D]} (\w(v)+1) \leq |V(G')| + \sum_{x \in D} \w(N_{G'}[x]) < n + 2kn.
        \]
    \end{claimproof}

    \begin{claim}\label{cl:br:tree}
        Steps (S1)--(S3) produce at most $n^{\Oh(k^2h^2)}$ instances, in total time $n^{\Oh(k^2h^2)}$.
    \end{claim}
    \begin{claimproof}
        Consider the search tree whose nodes are the instances produced by the procedure, where the children of an instance are the instances obtained from it by guessing in (S2) or (S3).
        Every guess deletes at least one vertex, so every root-to-leaf path has at most $n$ edges.
        For each guess of the color of a vertex $v$ in (S2), we call the edge to the child obtained for the exceptional color from \cref{cl:br:heavy} \emph{default} (if there is one); all other edges of the tree are \emph{marked}.

        Consider a root-to-leaf path.
        By \cref{cl:br:greedy}, it contains at most $M$ edges corresponding to (S3), each leading to one of at most $3^{2k}$ children.
        By \cref{cl:br:heavy}, each marked edge corresponding to (S2) decreases the potential by at least $n/2$.
        The potential is initially $0$ and never negative, and by \cref{cl:br:increase}, it increases only in (S3), by less than $(2k+1)n$ each time.
        Hence, the path contains fewer than $2(2k+1)M$ marked edges corresponding to (S2), each leading to one of at most $3$ children.
        A root-to-leaf path is determined by its length, the positions of its marked edges, and the children chosen at its marked edges.
        Thus there are at most $(n+1)^{(4k+3)M+1} \cdot 3^{2kM} \cdot 3^{(4k+2)M} = n^{\Oh(kM)} = n^{\Oh(k^2h^2)}$ root-to-leaf paths, and the tree has at most $n+1$ times as many nodes.
        Each node is processed in time $n^{2k+\Oh(1)}$, as every center has at most $2k$ vertices, which proves the claim.
    \end{claimproof}

    \medskip
    \noindent\textbf{The finishing step.}
    Consider an instance $(G',L')$ to which (S4) is applied, and let $C_1,\ldots,C_h$ be the first $h$ centers selected by the greedy selection; recall that they are pairwise anticomplete.

    \begin{claim}\label{cl:br:balanced}
        Let $G''$ be an induced subgraph of $G'$.
        For every $j \in [h]$, every component of $G'' - N_G(C_j)$ has at most $n/2$ vertices.
        Moreover, no vertex of $G''$ belongs to all the sets $N_G(C_j)$.
    \end{claim}
    \begin{claimproof}
        Fix $j$, and let $H$ be the graph of the instance in which $C_j$ was found in (S3).
        Then $N_H[C_j]$ is a $\frac{1}{2}$-balanced separator in $H$, and $V(G'') \subseteq V(G') \subseteq V(H) \setminus C_j$.
        Thus $G'' - N_G(C_j)$ is an induced subgraph of $H - N_H[C_j]$, so each of its components has at most $|V(H)|/2 \leq n/2$ vertices.
        Finally, a vertex in all the sets $N_G(C_j)$ would contradict \cref{cl:br:mult}, as no vertex of $G'$ is heavy.
    \end{claimproof}

    Now we make these separators pairwise non-crossing.
    We start with the family $\{(G',L')\}$, and consider the pairs $(j,j')$ with $1 \leq j < j' \leq h$ one by one.
    For the pair $(j,j')$, we replace each instance $(G_1,L_1)$ of the current family by the family obtained by applying \cref{thm:cut} to $G$, $(G_1,L_1)$, $A = C_j$, $B = C_{j'}$, $Z = V(G_1) \setminus (N_G(C_j) \cap N_G(C_{j'}))$, and $c = 4k$.
    The assumptions of \cref{thm:cut} are satisfied:
    the sets $C_j$, $C_{j'}$, and $Z$ are pairwise disjoint, as the centers are disjoint and consist of deleted vertices;
    the graph $G[C_j \cup C_{j'}]$ has at most $|C_j| + |C_{j'}| \leq 4k$ components;
    and $\dist_G(C_j,C_{j'} \mid Z) \geq 3$, as $C_j$ and $C_{j'}$ are anticomplete and no vertex of $Z$ has neighbors in both of them.
    Let $\cI'$ be the family obtained after considering all pairs.
    By the transitivity of equivalence, $\cI'$ is equivalent to $(G',L')$, and by \cref{thm:cut}, $\cI'$ consists of at most $(2n)^{24k^2h^2}$ instances and is computed in time $n^{\Oh(k^2h^2)}$.

    \begin{claim}\label{cl:br:noncross}
        Let $(G'',L'') \in \cI'$, and let $S_j = N_G(C_j) \cap V(G'')$ for $j \in [h]$.
        Then the sets $S_1,\ldots,S_h$ are pairwise non-crossing in $G''$.
    \end{claim}
    \begin{claimproof}
        Suppose that for some $j < j'$, there is a path $P$ in $G'' - (S_j \cap S_{j'})$ from a vertex $x \in S_j \setminus S_{j'}$ to a vertex $y \in S_{j'} \setminus S_j$.
        When considering the pair $(j,j')$, we applied \cref{thm:cut} to some instance $(G_1,L_1)$, and obtained an instance $(G_2,L_2)$ with $V(G'') \subseteq V(G_2) \subseteq V(G_1)$ and $\dist_G(C_j,C_{j'} \mid Z \cap V(G_2)) = \infty$, where $Z = V(G_1) \setminus (N_G(C_j) \cap N_G(C_{j'}))$.
        The vertices of $P$ belong to $V(G'')$ but not to $N_G(C_j) \cap N_G(C_{j'})$, so they belong to $Z \cap V(G_2)$.
        Thus extending $P$ by a neighbor of $x$ in $C_j$ and a neighbor of $y$ in $C_{j'}$ yields a $C_j$-$C_{j'}$-path in $G$ through $Z \cap V(G_2)$, a contradiction.
    \end{claimproof}

    Let $(G'',L'') \in \cI'$; we show that every component of $G''$ has at most $n/2$ vertices.
    This is clear if $|V(G'')| \leq n/2$, so assume otherwise, and let $\delta = n/(2|V(G'')|) < 1$.
    By \cref{cl:br:balanced,cl:br:noncross}, the sets $S_1,\ldots,S_h$ are pairwise non-crossing $\delta$-balanced separators in $G''$, and no vertex of $G''$ belongs to all of them.
    Hence, by \cref{thm:noncrossing}, every component of $G''$ has at most $\delta|V(G'')| = n/2$ vertices.
    The finishing step outputs $\cI'$.

    \medskip
    \noindent\textbf{Conclusion.}
    Every step replaces an instance by an equivalent family of subinstances, so by the transitivity of equivalence, the output family $\cI$ is equivalent to $(G,L)$.
    The instances output in (S1) and in the finishing step satisfy condition~\ref{it:branching:small}.
    By \cref{cl:br:tree}, steps (S1)--(S3) produce $n^{\Oh(k^2h^2)}$ instances in time $n^{\Oh(k^2h^2)}$, and for each of them, the finishing step takes time $n^{\Oh(k^2h^2)}$.
    As $h = \Oh(\log n)$, the total running time is $n^{\Oh(k^2 \log^2 n)}$.
\end{proof}

\section{Conclusion}

We showed that \lcol{3} can be solved in quasipolynomial time on every (quasi)tame class of graphs.
This does not resolve the main question: whether for tame classes the running time can be improved to polynomial.

\qpoly*

A particularly interesting special case is that of circle graphs, which form a tame class; see the introduction for the status of the problem in this class.

Recall that \MWIS and its generalizations, such as the $(\mathsf{tw} \leq r, \psi)$-\MWIS problem, are solved in tame classes using the framework of potential maximal cliques~\cite{BouchitteT01,FominTV15}.
It is unclear whether potential maximal cliques can be useful for \lcol{3}; in particular, all known algorithms for \lcol{3} in tame classes, as well as in $P_t$-free graphs, follow a different route.
This leads to the following meta-question.

\begin{question}\label{q:framework}
    Is there a natural problem that generalizes both $(\mathsf{tw} \leq r, \psi)$-\MWIS and \lcol{3}, and can be solved in (quasi)polynomial time on every tame class of graphs?
\end{question}

\paragraph{Acknowledgement.} The arguments leading to the proof were found with the help of ChatGPT. The presentation and exposition are due to the author. Claude was used to polish the write-up. The author takes full responsibility for the content of the paper.

\bibliographystyle{abbrv}
\bibliography{main}
\end{document}